\documentclass[11pt]{article}
\usepackage{paralist}
\usepackage{color}
\usepackage{bm}

\usepackage{soul}

\usepackage[cm]{fullpage}
\usepackage{algpseudocode,algorithm}
\usepackage{amsfonts}
\usepackage{amssymb}

\usepackage{amsmath}
\usepackage{constants}
\usepackage{amsthm}
\makeatletter
\newtheorem*{rep@theorem}{\rep@title}
\newcommand{\newreptheorem}[2]{%
	\newenvironment{rep#1}[1]{%
		\def\rep@title{#2 \ref{##1}}%
		\begin{rep@theorem}}%
		{\end{rep@theorem}}}
\makeatother

\newcommand{\junk}[1]{}

\usepackage{ifpdf}
\newcommand{\mydriver}{hypertex}
\ifpdf
\renewcommand{\mydriver}{pdftex}
\fi
\usepackage[breaklinks,\mydriver]{hyperref}

\usepackage[margin=1in]{geometry}

\theoremstyle{plain}
\newtheorem{theorem}{Theorem}[section]

\newtheorem{lemma}[theorem]{Lemma}

\newtheorem{claim}[theorem]{Claim}
\newtheorem{definition}[theorem]{Definition}

\theoremstyle{definition}

\newcommand{\II}{\ensuremath{\mathcal{I}}}

\newcommand{\pr}{{\ensuremath{\psi}}}
\newcommand{\appr}{\ensuremath{\Psi}}

\newcommand{\FF}{\ensuremath{\mathcal{F}}}
\newcommand{\disc}{\ensuremath{\mathrm{disc}}}
\newcommand{\Var}{\textrm{Var}}
\newcommand{\E}{\textrm{E}}
\newcommand{\OPT}{\textrm{OPT}}
\newcommand{\Property}{\ensuremath{{\Pi}}}
\newcommand{\Distri}{\ensuremath{\mathcal{D}}}
\newcommand{\Good}{\ensuremath{\mathcal{G}}}
\newcommand{\Perm}{\ensuremath{\Sigma}}

\newcommand{\poly}{\ensuremath{\mathrm{poly}}}
\newcommand{\dt}{\ensuremath{\mathrm{dt}}}
\newcommand{\Oracle}{\ensuremath{\mathcal{O}}}
\newcommand{\estim}{\ensuremath{\Psi}}
\newcommand{\Uniform}{\ensuremath{\mathcal{U}}}

\providecommand{\abs}[1]{\lvert#1\rvert} 

\title{Testable Bounded Degree Graph Properties Are Random Order Streamable}
\author{
	Morteza Monemizadeh\footnote{Department of Computer Science, Goethe-Universit\"{a}t Frankfurt, Germany. Partially supported by DFG grants ME 2088/3-(1/2) and ME 2088/4-1. Email: \url{monemi@ae.cs.uni-frankfurt.de}. 
	}
	\and
	S. Muthukrishnan\footnote{Rutgers University, Piscataway, NJ, USA. Email: \url{muthu@cs.rutgers.edu}.}
	\and
	Pan Peng\footnote{Faculty of Computer Science, University of Vienna, Austria. The research leading to these results has received funding from the European Research Council under the European Union's Seventh
		Framework Programme (FP/2007-2013) / ERC Grant Agreement
		no. 340506. Email: \url{pan.peng@univie.ac.at}. }
	\and
	Christian Sohler\footnote{Department of Computer Science, TU Dortmund, Germany. Supported by ERC Starting Grant 307696. Email: \url{christian.sohler@tu-dortmund.de}.}
}

\date{}

\begin{document}
\begin{titlepage}

\maketitle
\thispagestyle{empty}
\begin{abstract}
We study which property testing and sublinear time algorithms can be transformed into graph streaming algorithms for random order streams.
Our main result is that for bounded degree graphs, any property that is constant-query testable in the adjacency list model can be tested with 
\emph{constant space} in a single-pass in random order streams. Our result is obtained by estimating the distribution
of local neighborhoods of the vertices on a random order graph stream using constant space.

We then show that our approach can also be applied to constant time approximation algorithms for bounded degree graphs 
in the adjacency list model: As an example, we obtain a constant-space single-pass random order streaming algorithms for 
approximating the size of a maximum matching with additive error $\epsilon n$ ($n$ is the number of nodes). 

Our result establishes for the first time that a large class of sublinear algorithms can be simulated in random order streams, 
while $\Omega(n)$ space is needed for many graph streaming problems for adversarial orders. 
 \end{abstract}

\end{titlepage}

\section{Introduction}

Very large and complex networks abound. Some of the  prominent examples are 
gene regulatory networks, health/disease networks, and online social networks like Facebook, 
Google+, Linkedin and Twitter. 
The interconnectivity
of neurons in human brain, relations in database systems, and chip designs are some further examples.
Some of these networks can be quite large and it may be hard to store them completely in the main memory
and some may be too large to be stored at all. However, these networks contain valuable information
that we want to reveal. For example, social networks can provide insights into the structure of 
our society, and the structure in gene regulatory networks might yield insights into diseases. Thus, we need algorithms that can analyze the structure of these networks quickly.

One way to approach this problem is to design graph streaming algorithms~\cite{HRR99:stream, AMS96:space}. A graph streaming algorithm
gets access to a stream of edges in some order and exactly or approximately solves problems on the graph defined
by the stream. The challenge is that a graph streaming algorithm should use space sublinear in the
size of the graph. 
We will focus on algorithms that make only \emph{one pass} over the graph stream, unless we explicitly say otherwise.
It has been shown that many natural graph problems require $\Omega(n)$ space in the \emph{adversarial order} model 
where $n$ is the number of nodes in the graph and the edges can arrive in arbitrary order~(see eg.,\cite{FKMSZ05:graph,FKM08:distance}), 
and thus most of previous work has focused on the \emph{semi-streaming} model, in which the algorithms are allowed to use $O(n\cdot \poly\log n)$ space. 
However, in many interesting applications, the graphs are sparse and so they can
be fully stored in the semi-streaming model making this model useless in this setting. This raises
the question \emph{whether there are at least some natural conditions under which one can solve graph 
	problems with space $o(n)$, possibly even $\log^{O(1)} n$ or constant}.

One such condition that recently received increasing attention is that the edges arrive in \emph{random order}, i.e. in the
order of a uniformly random permutation of the edges (e.g., \cite{CCM08:lower,KMM12:matching,KKS14:matching}). Uniformly random or near-uniformly random ordering is a natural assumption and can arise in many contexts. Indeed, previous work has shown that some problems that are hard for adversarial streams can be solved in the random order model.
Konrad et al.~\cite{KMM12:matching} gave single-pass semi-streaming algorithms for maximum matching for bipartite and general graphs with approximation ratio strictly larger than $1/2$ in the random order semi-streaming model, while no such approximation algorithm is known in the adversary order model. Kapralov et al.~\cite{KKS14:matching} gave a polylogarithmic approximation algorithm in polylogarithmic space for estimating the size of maximum matching of an unweighted graph in one pass over a random order stream. Assadi et al.~\cite{AKL17:estimating} recently showed that in the adversarial order and \emph{dynamic} model where edges can be both inserted and deleted, any polylogarithmic approximation algorithm of maximum matching size requires $\tilde{\Omega}(n)$ space. On the other hand, Chakrabarti et al.~\cite{CCM08:lower} presented an $\Omega(n)$ space lower bound for any single pass algorithm 
for graph connectivity in the random order streaming model, which is very close to the optimal $\Omega(n\log n)$ space lower bound in the adversarial order model~\cite{SW15:tight}. In general, it is unclear which graph problems can be solved in random order streams using much smaller space than what is required for adversarially ordered streams.

An independent area of research is \emph{property testing}, where with certain \emph{query} access to an object (eg., random vertices or neighbors of a vertex for graphs), there are algorithms that can determine if the object satisfies a certain property, or is far from having such a property~\cite{RS96:robust,DBLP:journals/jacm/GoldreichGR98,GR02:testing}.
The area of property testing has seen fundamental results, including testing various general graph properties. For example, it has been shown that many interesting properties (including connectivity, planarity, minor-freeness, hyperfiniteness) of bounded degree graphs can be tested with a constant number of queries~\cite{GR02:testing,BSS10:minor,NS13:hyperfinite}. Another very related area of research is called \emph{constant-time} (or in general, \emph{sublinear-time}) \emph{approximation} algorithms, where we are given query access to an object (for example a graph) and the goal is to approximate the objective value of an
optimal solution. For example, in bounded degree graphs, one can approximate the cost of the optimal solution with constant query complexity for some
fundamental optimization problems (e.g., minimum spanning tree weight~\cite{DBLP:journals/siamcomp/ChazelleRT05}, maximal matching size~\cite{NO08:constant}; see also Section~\ref{sec:related_work}). 

A fundamental question is if such results from property testing and constant-time approximation algorithms will lead to better graph streaming algorithms. Huang and Peng~\cite{HP16:stream} recently considered the problem of estimating the minimum spanning tree weight and property testing for general graphs in dynamic and adversarial order model. They showed that a number of properties (e.g., connectivity, cycle-freeness) of general $n$-vertex graphs can be tested with space complexity $O(n^{1-\varepsilon})$ and one can $(1+\varepsilon)$-approximate the weight of minimum spanning tree with similar space guarantee. Furthermore, there exist $\Omega(n^{1-O(\varepsilon)})$ space lower bounds for these problems that hold even in the insertion-only model~\cite{HP16:stream}. 

\subsection{Overview of Results}

In this paper we provide a general framework that transforms bounded-degree graph property testing to very space-efficient random order streaming algorithms. 

%


To formally state our main result, we first review some basic definitions of graph property testing. A \emph{graph property} is a property that is invariant under graph isomorphism. Let $G=(V,E)$ be a graph with maximum degree upper bounded by a constant $d$, and we also call $G$ a \emph{$d$-bounded} graph. In the \emph{adjacency list model} for (bounded-degree) graph property testing, we are given query access to the adjacency list of the input $d$-bounded graph $G=(V,E)$. That is, for any vertex $v\in V$ and index $i\leq d$, one can query the $i$th neighbor (if exists) of vertex $v$ in constant time. Given a property $\Property$, we are interested in testing if a graph $G$ satisfies $\Property$ or is \emph{$\varepsilon$-far} from satisfying $\Property$ while making as few queries as possible, where $G$ is said to be $\varepsilon$-far from satisfying $\Property$ if one has to insert/delete more than $\varepsilon dn$ edges to make it satisfy $\Property$. We call a property \emph{constant-query testable} if there exists a testing algorithm (also called \emph{tester}) for this property such that the number of performed queries depends only on parameters $\varepsilon,d$ and is independent of the size of the input graph. 

Given a graph property $\Property$, we are interested in \emph{approximately} testing it in a single-pass stream with a goal similar to the above. That is, the algorithm uses little space and with high constant probability, it accepts the input graph $G$ if it satisfies $P$ and rejects $G$ if it is $\varepsilon$-far from satisfying $P$ (see Section~\ref{sec:property_testing} for formal definitions). Our main result is as follows.

\begin{theorem}\label{thm:const_query_testable_property}
	Any $d$-bounded graph property that is constant-query testable in the adjacency list model can be tested in the uniformly random order streaming model with constant space. 
\end{theorem}

To the best of our knowledge, this is the first non-trivial graph streaming algorithm with constant space complexity (measured in the number of \emph{words}, where a word is a space unit large enough to encode an ID of any vertex in the graph.) By the constructions in~\cite{HP16:stream}, there exist graph properties (e.g., connectivity and cycle-freeness) of $d$-bounded graphs such that any single-pass streaming algorithm in the insertion-only and \emph{adversary order} model must use $\Omega(n^{1-O(\varepsilon)})$ space. In contrast to this lower bound, our main result implies that $d$-bounded connectivity and cycle-freeness can be tested in constant space in the random order stream model, since they are constant-query testable in the adjacency list model~\cite{GR02:testing}. 

Our approach also works for simulating \emph{constant-time approximation} algorithms 
as graph streaming algorithms with constant space. For a minimization (resp., maximization) optimization problem $P$ and an instance $I$, we let $\OPT(I)$ denote the value of some optimal solution of $I$. We call a value $x$ an $(\alpha,\beta)$-approximation for the problem $P$, if for any instance $I$, it holds that $\OPT(I)\leq x\leq \alpha\cdot \OPT(I) +\beta$ (resp., $\frac{\OPT(I)}{\alpha} -\beta\leq x\leq \OPT(I)$).
For example, it is known that there exists a constant-query algorithm for $(1,\varepsilon n)$-approximating the maximal matching size of any $n$-vertex $d$-bounded graph~\cite{NO08:constant}. That is, the number of queries made by the algorithm is independent of $n$ and only depends on $\varepsilon,d$. As an application, we show: 

\begin{theorem}\label{thm:constant_time}
	Let $0<\varepsilon <1$ and $d$ be constants. Then there exists an algorithm that uses constant space in the random order model, 
	and with probability $2/3$, $(1,\varepsilon n)$-approximates the size of some maximal matching in $d$-bounded graphs.
\end{theorem}

We also remark that in a similar way, many other sublinear time algorithms for bounded degree graphs can be simulated in random order streams. 
Finally, our results can actually be extended to a model which requires weaker assumptions on the randomness of the order of edges in the stream, but we describe our results for the uniformly random order model, and leave the remaining details for later. 

\subsection{Technical Overview}

The local neighborhood of depth $k$ of a vertex $v$ is the subgraph rooted at $v$ and induced by all vertices of distance at most $k$ from $v$.
We call such a rooted subgraph a {\em $k$-disc}.
Suppose that we are given a sufficiently large graph $G$ whose maximum degree $d$ is constant. This means that for any constant $k$, a $k$-disc centered at an arbitrary vertex $v$ in $G$ has constant size. 
Now assume that there exists an algorithm $\mathcal{A}$ that, independent of the labeling of the vertices of $G$, accesses $G$ by querying random vertices and exploring their $k$-discs. We observe that any constant-query property tester (see for example \cite{GR11:proximity,CPS16:testing}) falls within the framework of such an algorithm. If instead of the graph $G$ we are given the distribution of $k$-discs of the vertices of $G$, we can use this distribution to simulate the algorithm $\mathcal{A}$ and output with high probability the same result as executing the algorithm $\mathcal{A}$ on $G$ itself.  
Thus, the problem of developing constant-query property testers in random order streams can be reduced to the problem of designing  streaming algorithms that approximate the distribution of $k$-discs in $G$.

The main technical contribution of this 
paper is an algorithm that given a random order stream $S$ of edges of an underlying $d$-bounded degree graph $G$, approximates the distribution of $k$-discs of $G$  up to an additive error of $\delta$. We would like to mention that if the edges arrive in adversarial order,  any algorithm
that approximates the distribution of $k$-discs of $G$ requires almost linear space \cite{VY11:streaming,HP16:stream}, 
hence the assumption of random order streams (or something similar) is necessary to obtain our result.

Now in order to approximate the distribution of $k$-discs of the graph $G$ we do the following. We proceed by sampling vertices uniformly at random and then perform a BFS for each sampled vertex using the arrival of edges along the stream $S$. Note that the new edges of the stream $S$ that do not connect to the currently explored vertices are discarded. Let us call the $k$-disc that is observed by doing such a BFS from some vertex $v$ to be $\Delta_1$. Due to possibility of missing edges during the BFS, this subgraph may be different from the true $k$-disc $\Delta_2$ rooted at $v$.

If we are allowed to use two passes of the stream, then one can collect the $k$-disc of $v$ in the first pass, and then verify if the collected disc is the true $k$-disc of $v$ in the second pass (see Section~\ref{sec:two_pass}). 
However, if we are restricted to a single pass, then it is more challenging to detect or verify if some edges have been missed in a collected disc. Fortunately, since the edges arrive in a uniformly random order, we can infer the conditional probability $\Pr[\Delta_1|\Delta_2]$. That is, given the true rooted subgraph $\Delta_2$, we can compute the conditional probability of seeing a rooted subgraph $\Delta_1$ in a random order stream when the true $k$-disc is $\Delta_2$.

We define the partial order on the set of $k$-discs given by $\Delta_1\preccurlyeq \Delta_2$ whenever $\Delta_1$
is a root-preserving isomorphic subgraph of $\Delta_2$. 
For every two $k$-discs $\Delta_1$ and $\Delta_2$ with $\Delta_1\preccurlyeq \Delta_2$ we compute  the conditional probability $\Pr[\Delta_1|\Delta_2]$. 
Using the set of all conditional probabilities $\Pr[\Delta_1|\Delta_2]$ we
can estimate or approximate the distribution of $k$-discs of the graph $G$ whose edges are revealed according to the stream $S$. 
In order to simplify the analysis of our algorithm, we require a natural independence condition for non-intersecting $k$-discs. Finally, we use the approximated distribution of $k$-discs to simulate the algorithm $\mathcal{A}$ by the machinery that we explained above. 

We remark that the idea of using a partial order to compute a distribution of
$k$-discs in bounded degree graphs has first been used in \cite{CPS16:testing}. However, the setting
in \cite{CPS16:testing} was quite different as it dealt with directed graphs where an edge can only be seen
from one side (and the sample sizes required in that paper were only slightly sublinear in $n$).

\subsection{Other Related Work}\label{sec:related_work}
Feigenbaum et al.~\cite{FKSV02:testing} initiated the study of property testing in streaming model, and they gave efficient testers for some properties of a sequence of data items (rather than graphs as we consider here). Bury and Schwiegelshohn~\cite{BS15:sublinear} gave a lower bound of $n^{1-O(\varepsilon)}$ on the space complexity of any algorithm that $(1-\varepsilon)$-approximates the size of maximum matching in adversarial streams. Kapralov et al.~\cite{KKS15:streaming} showed that in random streams, $\tilde{\Omega}(\sqrt{n})$ space is necessary to distinguish if a graph is bipartite or $1/2$-far from being bipartite. Previous work has extensively studied streaming graph algorithms in both the insertion-only and dynamic models, see the recent survey~\cite{McG14:stream}. 

In the framework of $d$-bounded graph property testing, it is now known that many interesting properties are constant-query testable in the adjacency list model, including $k$-edge connectivity, cycle-freeness, subgraph-freeness~\cite{GR02:testing}, $k$-vertex connectivity~\cite{YI08:vertex}, minor-freeness~\cite{HKNO09:local,BSS10:minor}, matroids related properties~\cite{ITY12:matroid,TY15:supermodular}, hyperfinite properties~\cite{NS13:hyperfinite}, subdivision-freeness~\cite{KY13:subdivision}. Constant-time approximation algorithms in $d$-bounded graphs are known to exist for a number of fundamental optimization problems, including $(1+\varepsilon)$-approximating the weight of minimum spanning tree~\cite{DBLP:journals/siamcomp/ChazelleRT05}, $(1,\varepsilon n)$-approximating the size of maximal/maximum matching~\cite{NO08:constant,YYI12:constant}, $(2,\varepsilon n)$-approximating the minimum vertex cover size~\cite{PR07:sublinear_distributed,MR09:distance,ORRR12:vetexcover}, $(O(\log d), \varepsilon n)$-approximating the minimum dominating set size~\cite{PR07:sublinear_distributed,NO08:constant}. For $d$-bounded minor-free graphs, there are constant-time $(1,\varepsilon n)$-approximation algorithms for the size of minimum vertex cover, minimum dominating set and maximum independent set~\cite{HKNO09:local}.

\section{Preliminaries}\label{sec:preliminaries}
Let $G=(V,E)$ be an $n$-vertex graph with maximum degree upper bounded by some constant $d$, where we often identify $V$ as $[n]:=\{1,\cdots,n\}$. We also call such a graph $d$-bounded graph. In this paper, we will assume the algorithms have the knowledge of $n,d$. We assume that $G$ is represented as a sequence of edges, which we denote as \textsc{Stream}$(G)$.

\subparagraph*{Graph $k$-discs.}
\vspace{-10pt}
Let $k\geq 1$. The $k$-disc around a vertex $v$ is the subgraph rooted at vertex $v$ and induced by the vertices within distance at most $k$ from $v$. Note that for an $n$-vertex graph, there are exactly $n$ $k$-discs. Let $\mathcal{H}_{d,k}=\{\Delta_1,\cdots,\Delta_N\}$ be the set of all $k$-disc isomorphism types, where $N=N_{d,k}$ is the number of all such types (and is thus a constant). In the following, we will refer to a $k$-disc of some vertex $v$ in the graph $G$ as $\disc_{k,G}(v)$ and a $k$-disc type as $\Delta$. Note that for every vertex $v$, there exists a unique $k$-disc type $\Delta\in\mathcal{H}_{d,k}$ such that $\disc_{k,G}(v)$ is isomorphic to $\Delta$, denoted as $\disc_{k,G}(v)\cong \Delta$. {(Throughout the paper, we call two rooted graphs $H_1,H_2$ isomorphic to each other if there is a root-preserving mapping from the vertex set of $H_1$ to the vertex set of $H_2$.)}

We further assume that all the elements in $\mathcal{H}_{d,k}$ are ordered according to the natural partial order among $k$-disc types. More specifically, for any two $k$-disc types $\Delta_i,\Delta_j$, we let $\Delta_i\succcurlyeq\Delta_j$ (or equivalently, $\Delta_j\preccurlyeq\Delta_i$) denote that $\Delta_j$ is root-preserving isomorphic to some subgraph of $\Delta_i$. Then we order all the $k$-disc types 
$\Delta_1,\cdots,\Delta_{N}$ such that if $\Delta_i\succcurlyeq\Delta_j$, then $i\leq j$. Let $\Good(j)$ denote all the indices $i$, except $j$ itself, such that $\Delta_i\succcurlyeq\Delta_j$.

\subparagraph*{Locally random order streams.}
\vspace{-10pt}
Let $\Perm_E$ denote the set of all permutations (or orderings) over the edge set $E$. Note that each $\sigma \in\Perm_E$ determines the order of edges arriving from the stream. Let $\Distri=\mathcal{D}(\Perm_E)$ denote a probability distribution over $\Sigma_E$. In particular, we let $\mathcal{U}=\mathcal{U}(\Sigma_E)$ denote the uniform distribution over $\Sigma_E$. 
%
%
Given a stream $\sigma$ of edges, 
we define the \emph{observed $k$-disc of $v$ from the stream}, denoted as $\disc_k(v,\sigma)$, to be the subgraph rooted at $v$ and induced by all edges that are sequentially collected from the stream and the endpoints of which are within distance at most $k$ to $v$. This is formally defined in the following algorithm \textsc{Stream\_$k$-disc}.

\begin{algorithm}[H]
	\caption{The observed $k$-disc of $v$ from the stream}~\label{alg:stream}
	\begin{algorithmic}[1]
		\Procedure{Stream\_$k$-disc}{\textsc{Stream}$(G)$,$k$,$v$}	
		\State $U\gets \{v\}$, $\ell_v=0$, $F \gets \emptyset$ 
		\For{$(u,w)\gets$ next edge in the stream}
		\If{exactly one of $u,w$, say $u$, is contained in $U$}
		\If{$\ell_u\leq k-1$}
		\State{$U\gets U\cup\{w\}, F\gets F\cup\{(u,w)\}$}
		\For{$x\in U$}
		\State{$\ell_x\gets$ the distance between $x$ and $v$ in the graph $G'=(U,F)$}
		\EndFor
		\EndIf
		\ElsIf{both $u,v$ are contained in $U$}
		\State{$F\gets F\cup\{(u,w)\}$}
		\For{$x\in U$}
		\State{$\ell_x\gets$ the distance between $x$ and $v$ in the graph $G'=(U,F)$}
		\EndFor
		\EndIf
		\EndFor
		\State \Return $\disc_k(v,\sigma)\gets$ the subgraph rooted at $v$ and induced by all edges in $F$
		\EndProcedure
	\end{algorithmic}
\end{algorithm}

Now we formally define a locally random distribution on the order of edges.
\begin{definition}~\label{def:local_distribution}
	Let $d,k>0$. Let $G=(V,E)$ be a $d$-bounded graph. Let $\Distri$ be a distribution over all the orderings of edges in $E$. Let $\Lambda_k=\{\lambda(\Delta_i|\Delta_j): 0\leq \lambda(\Delta_i|\Delta_j)\leq 1, \Delta_j\succcurlyeq \Delta_i, 1\leq i,j\leq N\}$ be a set of real numbers in $[0,1]$. We call $\Distri$ a \emph{locally random $\Lambda_k$-distribution} over $G$ with respect to $k$-disc types, if for $\sigma$ sampled from $\Distri$, the following conditions are satisfied:
	\begin{enumerate}
		\item (\emph{Conditional probabilities}) For any vertex $v$ with $k$-disc isomorphic to $\Delta_j$, the probability that its observed $k$-disc $\disc_{k}(v,\sigma)\cong \Delta_i$ is $\lambda(\Delta_i|\Delta_j)$, for any $i$ such that $\Delta_j\succcurlyeq\Delta_i$.  
		\item\label{def:local_item_2} (\emph{Independence of disjoint $k$-discs}) For any two disjoint $k$-discs $\disc_{k,G}(v)$ and $\disc_{k,G}(u)$, their observed $k$-discs $\disc_{k}(v,\sigma)$ and $\disc_{k}(u,\sigma)$ are independent.
	\end{enumerate}
\end{definition}

Note that the set $\Lambda_k$ cannot be an arbitrary set, as there might be no distribution satisfying the above condition. 
On the other hand, if there indeed exists a distribution satisfying the condition with numbers in $\Lambda_k$, then we call the set $\Lambda_k$ \emph{realizable}. In the following, we call a stream a \emph{locally random order stream} if there exists a family of realizable sets $\Lambda=\{\Lambda_k\}_{k\geq 1}$, such that the edge order is sampled from some locally random $\Lambda_k$-distribution with respect to $k$-disc types, for any integer $k\geq1$. We have the following lemma.
\begin{lemma}~\label{lemma:uniformly_random}
	Let $d\geq 1$. For any $k\geq 1$, there exists $n_0=n_0(k,d)$, such that for $n\geq n_0$, any $d$-bounded $n$-vertex graph $G=(V,E)$, the uniform permutation $\Uniform$ over $E$ is a locally random $\Lambda_k$-distribution over $G$ with respect to $k$-disc types, for some realizable $\Lambda_k:=\{\lambda(\Delta_i|\Delta_j): 0\leq \lambda(\Delta_i|\Delta_j)\leq 1, \Delta_j\succcurlyeq \Delta_i, 1\leq i,j\leq N\}$. Furthermore, if we let $\kappa:=\max_{i,j:\Delta_j\succcurlyeq\Delta_i}\frac{\lambda(\Delta_i|\Delta_j)}{\lambda(\Delta_i|\Delta_i)}$, $\lambda_{\min}:=\min_{i\leq N}\lambda(\Delta_i|\Delta_i)$, then $\kappa\leq 2^{2d^{k+1}}$, $\lambda_{\min}\geq \frac{1}{(2d^{k+1})!}$.
\end{lemma}
\begin{proof}
	Note that for any vertex $v$ with $\disc_{k,G}(v)\cong\Delta_j$, the probability that the observed $k$-disc of $v$ is isomorphic to $\Delta_i$ is exactly the fraction of orderings $\sigma$ such that $\disc_k(v,\sigma)\cong \Delta_i$, where $\Delta_j\succcurlyeq\Delta_i$. We use such a fraction, which is a fixed real number, to define $\lambda(\Delta_i|\Delta_j)$. Observe that for an ordering $\sigma$ sampled from $\Uniform$, it directly satisfies the second condition Item~\ref{def:local_item_2} in Definition~\ref{def:local_distribution}. Since there are at most $2d^{k+1}$ edges in any $k$-disc, the probability of observing a full $k$-disc is at least $\frac{1}{(2d^{k+1})!}$, that is, $\lambda_{\min}\geq \frac{1}{(2d^{k+1})!}$. Furthermore, since the $k$-disc type $\Delta_j$ might contain at most $\binom{|E(\Delta_j)|}{|E(\Delta_i)|}\leq 2^{2d^{k+1}}$ different subgraphs that are isomorphic to $\Delta_i$, it holds that $\lambda(\Delta_i|\Delta_j)\leq \sum_{\substack{F:F\textrm{ subgraph of }\Delta_j\\ F\cong\Delta_i}}\lambda(\Delta_i|\Delta_i)\leq 2^{2d^{k+1}}\lambda(\Delta_i|\Delta_i)$ for any $i,j$ such that $\Delta_j\succcurlyeq\Delta_i$. This completes the proof of the lemma.
\end{proof}
The above lemma shows that the uniformly random order stream is a special case of a locally random order stream. 
Another natural class of locally random order stream is $\ell$-wise independent permutation of edges for any $\ell=\omega_n(1)$ (i.e., any function that tends to infinity as $n$ goes to infinity) for $n$-vertex bounded degree graphs, but for our qualitative purposes here, it suffices to consider uniformly random order streams. 


\section{Approximating the $k$-Disc Type Distribution}\label{sec:approx_kdisc}
In this section, we show how to approximate the distribution of $k$-disc types of any $d$-bounded graph in locally random order streams. 

Recall that for any $k,d$, we let $N=N_{d,k}$ be the constant denoting the number of all possible $k$-disc isomorphism types. 
For any $i\leq N$, let $V_i$ be the set of vertices from $V$ with $k$-disc isomorphic to $\Delta_i$ in the input graph $G$, that is, $V_i:=\{v|v\in V,\disc_{k,G}(v)\cong \Delta_i\}$. Note that $f_i=\frac{|V_i|}{n}$ is the fraction of vertices with $k$-disc isomorphic to $\Delta_i$. 

\subsection{A Two-Pass Algorithm}\label{sec:two_pass}
We start with a discussion of a two-pass algorithm for approximating the distribution of $k$-disc types. 
The main idea is that in the first pass we can collect or observe the $k$-disc from any vertex $u$, and then in the second pass, we check if the observed $k$-disc is the true $k$-disc of $u$ or not. We can then use the statistics of the observed true $k$-discs to estimate the distribution of $k$-disc types. 

Slightly more formally, we first sample a large constant number of vertices and let $S$ denote the set of sampled vertices. Then in the first pass, for each vertex $u\in S$, we invoke the algorithm \textsc{Stream\_$k$-disc} to collect the observed $k$-disc of $u$, denoted as $H_u$, from the stream. In the second pass, for each vertex $w\in V(H_u)$, we collect all the incident edges to $w$. Then we let $H'_u$ denote the subgraph spanned by all edges (collected in the second pass) incident to vertices within distance at most $k$ to $u$. We check if $H_u$ is isomorphic to $H'_u$. It is not hard to see that the true $k$-disc of $u$ is observed if and only if $H_u$ is isomorphic to $H'_u$. For each $k$-disc type $\Delta_i$, we could then use the fraction of vertices $v$ in $S$ such that the true $k$-disc of $v$ is observed and is isomorphic to $\Delta_i$, to define an estimator for $f_i$. One should note that the naive estimator needs to be normalized appropriately by some probabilities and that there are dependencies between different variables, if one samples more than one starting vertex. Similar technical challenges also appear in our single-pass algorithm, for which we give detailed analysis in the following section. We omit further discussion on the two-pass algorithm here.

\subsection{A Single-Pass Algorithm}
In the following, we present our single-pass algorithm for approximating the distribution of $k$-disc types. We have the following lemma.
\begin{lemma}~\label{lemma:approx_k_disc_type}
	Let $G=(V,E)$ be a $d$-bounded graph presented in a locally random order stream defined by a $\Lambda_k$-distribution $\Distri$ over $G$ with respect to $k$-disc types, for some integer $k$. Let $\kappa:=\max_{i,j:\Delta_j\succcurlyeq\Delta_i}\frac{\lambda(\Delta_i|\Delta_j)}{\lambda(\Delta_i|\Delta_i)}$, $\lambda_{\min}:=\min_{i\leq N}\lambda(\Delta_i|\Delta_i)$. Then for any constant $\delta>0$, there exists a single-pass streaming algorithm that uses $O(\frac{\kappa^{2N}\cdot d^{3k+2}\cdot3^{3N+1}}{\delta^2\lambda_{\min}})$ space, and with probability $\frac{2}{3}$,  for any $i\leq N$, approximates the fraction $f_i$ of vertices with $k$-disc isomorphic to $\Delta_i$ in $G$ with additive error $\delta$.
\end{lemma}

\begin{proof}
	Our algorithm is as follows. We first sample a constant number of vertices, which are called centers. Then for each center $v$, we  collect the observed $k$-disc of $v$ from the stream. 
	Then we postprocess all the collected edges and use the corresponding empirical distribution of $k$-disc types of all centers to estimate the distribution of $k$-disc types of the input graph. The formal description is given in Algorithm~\ref{alg:approx_k_disc}.
	
	\begin{algorithm}[h]
		\caption{Approximating the distribution of $k$-disc types}~\label{alg:approx_k_disc}
		\begin{algorithmic}[1]
			\Procedure{$k$-disc\_distribution}{\textsc{Stream}$(G)$,$\Lambda_k$,$n,d,k,\delta$}	
			\State sample a set $A$ of $s:=\frac{8\kappa^{2N}\cdot d^{2k+1}\cdot3^{3N+1}}{\delta^2\lambda_{\min}}$ vertices uniformly at random 
			\For{each $v\in A$}
			\State $H_v\gets$ \textsc{Stream\_$k$-disc}(\textsc{Stream}($G$),$v$,$k$) \Comment{to collect observed $k$-disc of $v$}
			\EndFor
			\EndProcedure
			\State 
			\Procedure{Postprocessing}{}
			\State $H\gets$ the graph spanned by $\cup_{v\in A}H_v$
			\For{$i=1$ to $N$}
			\State $Y_i\gets |\{v:v\in A, \disc_{k,H}(v)\cong\Delta_i\}|/s$
			\State $X_i\gets (Y_i-\sum_{j\in \mathcal{G}(i)}X_j\cdot\lambda(\Delta_i|\Delta_j))\cdot\lambda^{-1}(\Delta_i|\Delta_i)$.
			\EndFor
			\State \Return $X_1,\cdots,X_{N}$
			\EndProcedure
		\end{algorithmic}
	\end{algorithm}
	
	Note that since there are $s=\frac{8\kappa^{2N}\cdot d^{2k+1}\cdot3^{3N+1}}{\delta^2\lambda_{\min}}$ vertices in $A$ and only edges that belong to the $k$-discs of these vertices will be collected by our algorithm, the space complexity of the algorithm is $O(s d^{k+1})=O(\frac{\kappa^{2N}\cdot d^{3k+2}\cdot3^{3N+1}}{\delta^2\lambda_{\min}})$, which is constant.
	
	
	Now we show the correctness of the algorithm.  
	
	We let $A\sim\Uniform_V$ denote that $A$ is the set of $s$ vertices sampled uniformly at random from $V$. For any $i\leq N$, let
	$A_i$ be the set of vertices from $A$ with $k$-disc isomorphic to $\Delta_i$ in the input graph $G$, that is, $A_i:=\{v|v\in A,\disc_{k,G}(v)\cong \Delta_i\}$. Note that $\E_{A\sim\Uniform_V}[|A_i|]=s\cdot\frac{|V_i|}{n}$. 
	
	Let $\beta_i=3^{i-N-2}, \theta_i=(3\kappa)^{i-N-1}$. By Chernoff bound and our setting of $s$ which satisfy that $s\geq \Omega(\frac{1}{(\delta\theta_i)^2\beta_i})$, we have the following claim. 
	\begin{claim}
		For any $i\leq N$, $\Pr_{A\sim\Uniform_V}[\abs{\frac{|A_i|}{s}-\frac{|V_i|}{n}}\leq \delta\theta_i]\geq 1-\beta_i$.
	\end{claim}
	
	We assume for now that $A$ is a fixed set with $s$ vertices. 
	%
	We let $\sigma\sim \Distri$ denote that the edge ordering $\sigma$ is sampled from $\Distri$. 
	For any $v\in A$, let $Z_{v,i}$ be the indicator random variable of the event that the observed $k$-disc $\disc_k(v,\sigma)$ of $v$ is isomorphic to $\Delta_i$ for $\sigma\sim\Distri$. Note that $\Pr_{\sigma\sim\Distri}[Z_{v,i}=1]=\lambda(\Delta_i|\Delta_j)$ if $\disc_{k,G}(v)\cong\Delta_j$. Let $Y_i^{(\sigma)}:=\frac{|\{v:v\in A,\disc_k(v,\sigma)\cong\Delta_i\}|}{s}$ denote the fraction of vertices in $A$ with observed $k$-disc isomorphic to $\Delta_i$. By definition, it holds that $Y_i^{(\sigma)}=\frac{1}{s}\sum_{\substack{v\in A_j\\ j\in\Good(i)\cup\{i\}}} Z_{v,i}$, and furthermore, 
	$\E_{\sigma\sim \Distri}[Y_i^{(\sigma)}] = \frac{1}{s}\sum_{j\in \Good(i)\cup\{i\}}|A_j|\cdot\lambda(\Delta_i|\Delta_j)$. Let $X_i^{(\sigma)}= (Y_i^{(\sigma)}-\sum_{j\in \mathcal{G}(i)}X_j^{(\sigma)}\cdot\lambda(\Delta_i|\Delta_j))\cdot\lambda^{-1}(\Delta_i|\Delta_i)$. 
	
	We have the following claim.
	\begin{claim}~\label{claim:expectation}
		For any $i\leq N$, it holds that $\E_{\sigma\sim \Distri}[X_i^{(\sigma)}]=\frac{|A_i|}{s}$. 
	\end{claim}
\begin{proof}
	We prove the claim by induction. For $i=1$, it holds that $\E_{\sigma\sim \Distri}[X_1^{(\sigma)}]=\E_{\sigma\sim \Distri}[Y_1^{(\sigma)}]\cdot\lambda^{-1}(\Delta_1|\Delta_1)=\frac{|A_1|}{s}\cdot \lambda(\Delta_1|\Delta_1)\cdot \lambda^{-1}(\Delta_1|\Delta_1)=\frac{|A_1|}{s}$. Assuming that the claim holds for $i-1$, and we prove it holds for $i$ as well. By definition, we have that
	\begin{eqnarray*}
		&&\E_{\sigma\sim \Distri}[X_i^{(\sigma)}]=\E_{\sigma\sim \Distri}[(Y_i^{(\sigma)}-\sum_{j\in \mathcal{G}(i)}X_j^{(\sigma)}\cdot \lambda(\Delta_i|\Delta_j))\cdot\lambda^{-1}(\Delta_i|\Delta_i)]\\
		&=& \Big(\sum_{j\in\Good(i)\cup\{i\}}\frac{|A_j|}{s}\cdot\lambda(\Delta_i|\Delta_j) -\sum_{j\in \mathcal{G}(i)}\E_{\sigma\sim \Distri}[X_j^{(\sigma)}]\cdot \lambda(\Delta_i|\Delta_j)\Big)\cdot \lambda^{-1}(\Delta_i|\Delta_i)\\	
		&=&\Big(\sum_{j\in\Good(i)\cup\{i\}}\frac{|A_j|}{s}\cdot\lambda(\Delta_i|\Delta_j) -\sum_{j\in \mathcal{G}(i)}\frac{|A_j|}{s}\cdot \lambda(\Delta_i|\Delta_j)\Big)\cdot \lambda^{-1}(\Delta_i|\Delta_i)
		=\frac{|A_i|}{s},
	\end{eqnarray*}
	where the second to last equation follows from the induction. 
\end{proof}
	
	We can now bound the variance of $Y_i^{(\sigma)}$ as shown in the following claim. 
	\begin{claim}~\label{claim:variance}
		For any $i\leq N$, it holds that $\Var_{\sigma\sim \Distri}[Y_i^{(\sigma)}]\leq \frac{1}{s^2} \cdot d^{2k+1} \sum_{j\in \Good(i)\cup\{i\}}|A_j|\cdot\lambda(\Delta_i|\Delta_j)$.
	\end{claim}
	\begin{proof}
		Recall that $Y_i^{(\sigma)}=\frac{1}{s}\sum_{\substack{v\in A_j\\ j\in\Good(i)\cup\{i\}}} Z_{v,i}$. Note that for each $v\in A$, by the independence assumption on $\Distri$, the random variable $Z_{v,i}$ can only correlate with the corresponding variables for vertices that are within distance at most $2k$ from $v$. The number of such vertices is at most $1+d+d^2+\cdots+d^{2k}<d^{2k+1}$. Let $\dt(u,v)$ denote the distance between $u,v$ in the graph $G$. Then we have that
		\begin{eqnarray*}
			&&\E_{\sigma\sim \Distri}[(\sum_{\substack{v\in A_j\\j\in\Good(i)\cup\{i\} }}Z_{v,i})^2] = \E_{\sigma\sim \Distri}[\sum_{\substack{v\in A_j\\j\in\Good(i)\cup\{i\} }}\sum_{\substack{u\in A_j\\j\in\Good(i)\cup\{i\} }}Z_{v,i}\cdot Z_{u,i}]\\
			&=& \E_{\sigma\sim \Distri}[\sum_{\substack{v\in A_j\\j\in\Good(i)\cup\{i\}}} (\sum_{\substack{u\in A_j\\j\in\Good(i)\cup\{i\}\\ \dt_G(u,v)\leq 2k}}Z_{v,i}\cdot Z_{u,i} +
			\sum_{\substack{u\in A_j\\j\in\Good(i)\cup\{i\}\\ \dt_G(u,v)> 2k}}Z_{v,i}\cdot Z_{u,i})]\\
			&\leq &\E_{\sigma\sim \Distri}[\sum_{\substack{v\in A_j\\j\in\Good(i)\cup\{i\}}} \sum_{\substack{u\in A_j\\j\in\Good(i)\cup\{i\}\\ \dt_G(u,v)\leq 2k}}Z_{v,i}] + \left(\sum_{\substack{j\in\Good(i)\cup\{i\}}}
			[|A_j|]\cdot\lambda(\Delta_i|\Delta_j)\right)^2\\
			&\leq &d^{2k+1}\E_{\sigma\sim \Distri}[\sum_{\substack{v\in A_j\\j\in\Good(i)\cup\{i\}}} Z_{v,i}] + (\E_{\sigma\sim \Distri}[\sum_{\substack{v\in A_j\\j\in\Good(i)\cup\{i\}}}Z_{v,i}])^2\\
			&=&d^{2k+1} \cdot\sum_{j\in \Good(i)\cup\{i\}}|A_j|\cdot\lambda(\Delta_i|\Delta_j)+ (\E_{\sigma\sim \Distri}[\sum_{\substack{v\in A_j\\j\in\Good(i)\cup\{i\}}}Z_{v,i}])^2,
		\end{eqnarray*}
		where the first inequality follows from the fact that $Z_{u,i}\leq 1$, and that for any two vertices $u,v$ with $\dt(u,v)>2k$, $Z_{u,i},Z_{v,i}$ are independent. 
		%

		Then we have that
		\begin{eqnarray*}
			&&\Var_{\sigma\sim \Distri}[Y_i^{(\sigma)}]
			=\frac{1}{s^2}\cdot \Var_{\sigma\sim \Distri}[\sum_{\substack{v\in A_j\\j\in\Good(i)\cup\{i\} }}Z_{v,i}] \\
			&=&\frac{1}{s^2}\left(\E_{\sigma\sim \Distri}[(\sum_{\substack{v\in A_j\\j\in\Good(i)\cup\{i\} }}Z_{v,i})^2]-(\E_{\sigma\sim \Distri}[\sum_{\substack{v\in A_j\\j\in\Good(i)\cup\{i\} }}Z_{v,i}])^2\right)\\
			&\leq & \frac{1}{s^2} \cdot d^{2k+1} \sum_{j\in \Good(i)\cup\{i\}}|A_j|\cdot\lambda(\Delta_i|\Delta_j).
		\end{eqnarray*} 
		%
	\end{proof}
	
	We next prove that each $X_i^{(\sigma)}$ is concentrated around its expectation with high probability.
	\begin{claim}~\label{claim:x_i_concentration}
		For any $i\leq N$, it holds that $\Pr_{\sigma\sim \Distri}[|X_i^{(\sigma)}-\E_{\sigma\sim \Distri}[X_i^{(\sigma)}]|\leq \theta_i\delta]\geq 1-\beta_i$.
	\end{claim}
	\begin{proof}
		We prove the claim by induction. For $i=1$, it holds that
		\begin{eqnarray*}
			&&\Pr_{\sigma\sim \Distri}[|X_1^{(\sigma)}-\E_{\sigma\sim \Distri}[X_1^{(\sigma)}]|\leq \theta_1\delta]\leq\Pr_{\sigma\sim \Distri}[|Y_1^{(\sigma)}-\E_{\sigma\sim \Distri}[Y_1^{(\sigma)}]|\cdot\lambda^{-1}(\Delta_1|\Delta_1)\geq \delta\theta_1]\\
			&\leq& \frac{\Var_{\sigma\sim \Distri}[Y_1^{(\sigma)}]}{(\delta\theta_1)^2\cdot\lambda^2(\Delta_1|\Delta_1)} 
			\leq \frac{d^{2k+1}|A_1|\cdot \lambda(\Delta_1|\Delta_1)}{s^2\cdot (\delta\theta_1)^2\cdot\lambda^2(\Delta_1|\Delta_1)}
			\leq \frac{d^{2k+1}}{s (\delta\theta_1)^2\cdot\lambda(\Delta_1|\Delta_1)}
			\leq \beta_1,
		\end{eqnarray*}
		where the last inequality follows from our choice of $\beta_1,\theta_1$ and $s$ which satisfy that $s\geq \frac{d^{2k+1}}{ (\delta\theta_1)^2\beta_1\cdot\lambda(\Delta_1|\Delta_1)}$.
		Now let us consider arbitrary $i\geq 2$, assuming that the claim holds for any $j\leq i-1$. First, with probability (over the randomness that $\sigma\sim \Distri$) at least $1-\sum_{j=1}^{i-1}\beta_j = 1-\sum_{j=1}^{i-1}3^{j-N-2}\geq 1-\frac{\beta_i}{2}$, it holds that for all $j\leq i-1$, 
		$|X_j^{(\sigma)}-\E_{\sigma\sim \Distri}[X_j^{(\sigma)}]|\leq \theta_j\delta$. This further implies that with probability at least $1-\frac{\beta_i}{2}$, 
		\begin{eqnarray*}
			&&|\sum_{j\in \mathcal{G}(i)}X_j^{(\sigma)}\cdot\frac{\lambda(\Delta_i|\Delta_j)}{\lambda(\Delta_i|\Delta_i)}-\E_{\sigma\sim \Distri}[(\sum_{j\in \mathcal{G}(i)}X_j^{(\sigma)}\cdot\frac{\lambda(\Delta_i|\Delta_j)}{\lambda(\Delta_i|\Delta_i)})]|\\
			&\leq& \sum_{j\in \mathcal{G}(i)}|X_j^{(\sigma)}-\E_{\sigma\sim \Distri}[X_j^{(\sigma)}]|\cdot \frac{\lambda(\Delta_i|\Delta_j)}{\lambda(\Delta_i|\Delta_i)}\\
			&\leq& \sum_{j\in \mathcal{G}(i)}\delta\theta_j\cdot \frac{\lambda(\Delta_i|\Delta_j)}{\lambda(\Delta_i|\Delta_i)}
			\leq \kappa\cdot\sum_{j\in \mathcal{G}(i)}\delta\theta_j 
			\leq\kappa\cdot\sum_{j=1}^{i-1}\delta(3\kappa)^{j-N}\leq \frac{\theta_i\delta}{2}.
		\end{eqnarray*}

		Now note that 
		\begin{eqnarray*}
			&&\Pr_{\sigma\sim\Uniform}[|Y_i^{(\sigma)}-\E[Y_i^{(\sigma)}]|\cdot \lambda(\Delta_i|\Delta_i)^{-1}\geq \frac{\theta_i\delta}{2}]
			\leq \frac{4\cdot\Var_{\sigma\sim\Distri}[Y_i^{(\sigma)}]}{(\delta\theta_i)^2\cdot\lambda(\Delta_i|\Delta_i)^2}\\
			&\leq& \frac{4\cdot d^{2k+1} \sum_{j\in \Good(i)\cup\{i\}}|A_j|\cdot\lambda(\Delta_i|\Delta_j)}{s^2\cdot (\delta\theta_i)^2\cdot\lambda(\Delta_i|\Delta_i)^2}
			\leq \frac{4\cdot d^{2k+1} \cdot\kappa}{s\cdot (\delta\theta_i)^2\cdot\lambda(\Delta_i|\Delta_i)}
			\leq
			\frac{\beta_i}{2}, 
		\end{eqnarray*}
		where the last inequality follows from our choice of
		$\beta_i,\theta_i$ and $s$ which satisfy that $s\geq \frac{8\kappa\cdot d^{2k+1}}{ (\delta\theta_i)^2\beta_i\cdot\lambda(\Delta_i|\Delta_i))}$.
		%
		%
		
		Therefore, with probability (over $\sigma\sim \Distri$) at least $1-\frac{\beta_i}{2}-\frac{\beta_i}{2}=1-\beta_i$, it holds that 
		\begin{eqnarray*}
			&&|X_i^{(\sigma)}-\E_{\sigma\sim\Distri}[X_i^{(\sigma)}]|\\
			&=& \left|\frac{Y_i^{(\sigma)}-\sum_{j\in \mathcal{G}(i)}X_j^{(\sigma)}\cdot\lambda(\Delta_i|\Delta_j)}{\lambda(\Delta_i|\Delta_i)}-\E_{\sigma\sim\Distri}\left[\frac{Y_i^{(\sigma)}-\sum_{j\in \mathcal{G}(i)}X_j^{(\sigma)}\cdot\lambda(\Delta_i|\Delta_j)}{\lambda(\Delta_i|\Delta_i)}\right]\right|\\
			&=&\left|\frac{(Y_i^{(\sigma)}-\E_{\sigma\sim\Distri}[Y_i^{(\sigma)}])}{\lambda(\Delta_i|\Delta_i)}-\left(\sum_{j\in \mathcal{G}(i)}X_j^{(\sigma)}\cdot\frac{\lambda(\Delta_i|\Delta_j)}{\lambda(\Delta_i|\Delta_i)}-\E_{\sigma\sim\Distri}[(\sum_{j\in \mathcal{G}(i)}X_j^{(\sigma)}\cdot\frac{\lambda(\Delta_i|\Delta_j)}{\lambda(\Delta_i|\Delta_i)})]\right)\right|\\
			&\leq & \frac{\delta\theta_i}{2} +\frac{\delta\theta_i}{2} = \delta\theta_i.
		\end{eqnarray*}	
	\end{proof}

	Now with probability (over both $A\sim \Uniform_V$ and $\sigma\sim\Distri$) at least $1-\beta_i-\beta_i$, it holds that
	\begin{eqnarray*}
		&&\left| X_i^{(\sigma)}-\frac{|V_i|}{n}\right|\leq \left|X_i^{(\sigma)}-\E_{\sigma\sim\Distri}[X_i^{(\sigma)}]\right| + \left|{\E_{\sigma\sim\Distri}[X_i^{(\sigma)}]-\frac{|V_i|}{n}}\right| \\
		&=& \left|{X_i^{(\sigma)}-\E_{\sigma\sim\Distri}[X_i^{(\sigma)}]}\right| +\left|{\frac{|A_i|}{s}-\frac{|V_i|}{n}}\right|
		\leq \delta\theta_i+\delta\theta_i = 2\delta\theta_i.
	\end{eqnarray*}
	
	Finally, with probability at least $1-2\sum_{j=1}^N\beta_j=1-2\sum_{j=1}^N3^{j-N-2}\geq 1-\frac{1}{3}$, it holds that for all $i\leq N$, $|X_i-\frac{|V_i|}{n}|\leq 2\theta_i\delta\leq \delta$. This completes the proof of the lemma.
\end{proof}

\section{Constant-Space Property Testing}\label{sec:property_testing}
In this section, we show how to transform constant-query property testers in the adjacency list model to constant-space property testers in the random order stream model in a single pass and prove our main result Theorem~\ref{thm:const_query_testable_property}. (Our transformation also works in the locally random order model as defined in Definition~\ref{def:local_distribution}, but for simplicity, we only state our result in the uniformly random order model.) 
\begin{definition}
	Let $\Pi=(\Pi_n)_{n\in \mathbb{N}}$ be a property of $d$-bounded graphs, where $\Pi_n$ is a property of graphs with $n$ vertices. We say that $\Pi$ is testable with query complexity $q$, if for every $\varepsilon,d$ and $n$, there exists an algorithm that performs $q=q(n,\varepsilon,d)$ queries to the adjacency list of the graph, and with probability at least $2/3$, accepts any $n$-vertex $d$-bounded graph $G$ satisfying $\Pi$, and rejects any $n$-vertex $d$-bounded graph that is $\varepsilon$-far from satisfying $\Pi$. If $q=q(\varepsilon,d)$ is a function independent of $n$, then we call $\Pi$ constant-query testable.
\end{definition}

Similarly, we can define constant-space testable properties in graph streams. 

\begin{definition}
	Let $\Pi=(\Pi_n)_{n\in \mathbb{N}}$ be a property of $d$-bounded graphs, where $\Pi_n$ is a property of graphs with $n$ vertices. We say that $\Pi$ is testable with space complexity $q$, if for every $\varepsilon,d$ and $n$, there exists an algorithm that performs a single pass over an edge stream of an $n$-vertex $d$-bounded graph $G$, uses $q=q(n,\varepsilon,d)$ space, and with probability at least $2/3$, accepts $G$ if it satisfies $\Pi$, and rejects $G$ if it is $\varepsilon$-far from satisfying $\Pi$. If $q=q(\varepsilon,d)$ is a function independent of $n$, then we call $\Pi$ constant-space testable.
\end{definition}



The proof of Theorem~\ref{thm:const_query_testable_property} is based on the following known fact: every constant-query property tester can be simulated by some canonical tester which only samples a constant number of vertices, and explores the $k$-discs of these vertices, and then makes deterministic decisions based on the explored subgraph. This implies that it suffices to approximate the distribution of $k$-disc types of the input graph to test the corresponding property. Formally, we will use the following lemma relating the constant-time testable properties and their $k$-disc distributions. For any graph $G$, let $S_{G,k}$ denote the subgraph spanned by the union of $k$-discs rooted at $k$ uniformly sampled vertices from $G$. The following lemma is implied by Lemma 3.2 in \cite{CPS16:testing} (which was built on~\cite{GT03:three} and~\cite{GR11:proximity}). ({The result in~\cite{CPS16:testing} is stated for $d$-bounded directed graphs, while it also holds in the undirected case.})
\begin{lemma}~\label{lemma:constant_query}
	Let $\Property=(\Property_n)_{n\in \mathbb{N}}$ be any $d$-bounded graph property that is testable with $q=q(\varepsilon,d)$ query complexity in the adjacency list model. Then there exist integer $n_0$, $k=c\cdot q$ for some large universal constant $c$, and an infinite sequence of $\mathcal{F}=\{\mathcal{F}_n\}_{n\geq n_0}$ such that for any $n\geq n_0$, $\mathcal{F}_n$ is a set of graphs, each being a union of $k$ disjoint $k$-discs, and for any $n$-vertex graph $G$,
	\begin{itemize}
		\item if $G$ satisfies $\Property_n$, then with probability \emph{at most} $\frac{5}{12}$, $S_{G,k}$ is isomorphic to one of the members in $\mathcal{F}_n$. 
		\item if $G$ is $\varepsilon$-far from satisfying $\Property_n$, then with probability \emph{at least} $\frac{7}{12}$, $S_{G,k}$ is isomorphic to one of the members in $\mathcal{F}_n$. 
	\end{itemize} 
\end{lemma}

Now we are ready to give the proof of Theorem~\ref{thm:const_query_testable_property}, which follows almost directly from the proof of Theorem 1.1 in~\cite{CPS16:testing}. For the sake of completeness, we present the full proof here.
\begin{proof}[Proof of Theorem~\ref{thm:const_query_testable_property}]
	
Let $\Pi=(\Pi_n)_{n\in \mathbb{N}}$ be any property that is testable with query complexity $q=q(\varepsilon,d)$ in the adjacency list model. We set $k=c\cdot q$ and let $\FF_n$ be the set of graphs as guaranteed in Lemma~\ref{lemma:constant_query}. Note that each subgraph $F=(\Gamma_1,\cdots,\Gamma_k)\in\FF_n$ is a multiset of $k$-discs. Set $N=N(d,k)$, $\delta=\frac{1}{48(2kN)^k}$. Let $\Lambda_k$ be the set of probabilities as guaranteed in Lemma~\ref{lemma:uniformly_random}. Let $n_1:=n_1(d,k)$ be some sufficiently large constant.

Now let us describe our random order streaming algorithm for testing if an $n$-vertex $d$-bound graph $G$ satisfies $\Pi_n$ or is $\varepsilon$-far from satisfying $\Pi_n$. If $n<n_1$, we trivially test $\Pi_n$ with constant space by storing the whole graph which contains at most $O(dn_1)$ edges. If $n\geq n_1$, we first invoke the algorithm \textsc{$k$-Disc\_Distribution}(\textsc{Stream}($G$), $\Lambda_k,n,d,k,\delta$) to get estimators $X_1,\cdots,X_N$ for the fraction $f_1,\cdots, f_N$ of vertices whose $k$-discs are isomorphic to $\Delta_1,\cdots,\Delta_N$, respectively. Then for each $F=(\Gamma_1,\cdots,\Gamma_k)\in \FF_n$, we calculate its \emph{empirical frequency} as $\estim(F)=\frac{\prod_{i=1}^N\binom{X_i\cdot n}{x_i}}{\binom{n}{k}}$, where $x_i$ is the number of copies among $\Gamma_1,\cdots,\Gamma_k$ that are of the same type as $\Delta_i$, for $1\leq i\leq N$. Finally, we accept the graph if and only if $\sum_{F\in\mathcal{F}_n}\estim(F)<\frac{1}{2}$.
	
Note that the space used by the algorithm is a constant. More precisely, the space complexity is $O(\max\{dn_1, \allowdisplaybreaks \frac{\kappa^{2N}\cdot d^{3k+2}\cdot3^{3N+1}}{\delta^2\lambda_{\min}}\})\allowdisplaybreaks =O(\max\{dn_1, \allowdisplaybreaks \frac{\kappa^{2N}\cdot d^{3k+2}\cdot3^{3N+1}\cdot(2kN)^{2k}}{\lambda_{\min}}\})$, where the equation follows from Lemma~\ref{lemma:approx_k_disc_type} and our setting of $\delta$.

Now we show the correctness of the algorithm. Note that we only need to consider the case that $n\geq n_1$. By Lemma~\ref{lemma:approx_k_disc_type}, with probability at least $2/3$, it holds that for any $i\leq N$, $|X_i-f_i|\leq \delta$. In the following, we will condition on this event and we will prove that 
	\begin{itemize}
	\item if $G$ satisfies $\Pi_n$, then $\sum_{F \in \FF_n} \appr(F) < \frac12$, and
	\item if $G$ is $\varepsilon$-far from satisfying $\Pi_n$, then $\sum_{F \in \FF_n} \appr(F) \ge \frac12$.
	\end{itemize}
 This would complete the proof.

	
	
	
	For every $F = \{\Gamma_1,\dots,\Gamma_k\} \in \FF_n$ and the relevant $x_1, \dots, x_N$, we will study $\pr(\Gamma_1, \dots, \Gamma_{k}) := \frac{\prod_{i=1}^N \binom{f_i\cdot n}{x_i}}{\binom{n}{k}}$, from which we will obtain the required bounds for $\sum_{F \in \FF_n} \appr(F)$.
	
	Observe that for any multiset $\{\Gamma_1, \dots, \Gamma_{k}\}$, the probability that the $k$-discs of $k$ vertices sampled uniformly at random (without replacement) span a subgraph isomorphic to the subgraph corresponding to $\{\Gamma_1, \dots, \Gamma_{k}\}$ has the \emph{multivariate hypergeometric distribution} with parameters $n, f_1\cdot n,\allowbreak\dots, f_N\cdot n, k$. That is, if for every $i\leq N$, there are exactly $x_i$ copies in the multiset $\{\Gamma_1, \dots, \Gamma_{k}\}$ that are of the same isomorphic type as $\Gamma_i$ (note that $x_1 + \dots + x_N = k$ 
	for any $1 \le i \le N$), then the probability that the subgraph $S_{G,k}$ spanned by $k$-discs of $k$ uniformly sampled vertices is isomorphic to $\{\Gamma_1, \dots, \Gamma_{k}\}$ is equal to $\pr(\Gamma_1, \dots, \Gamma_{k}) = \frac{\prod_{i=1}^N \binom{f_i\cdot n}{x_i}}{\binom{n}{k}}$, where we assumed $\binom{L}{M} = 0$ for $L < M$.
	
	To study the relation between $\appr(F)$ and $\pr(F)$, we begin with the following auxiliary claim.
	
	\begin{claim}
		\label{claim:binomial}
		For any $i$, if $|X_i - f_i| \le \delta $, it holds that $|\binom{X_i\cdot n}{x_i} - \binom{f_i\cdot n}{x_i}| \le 4 \delta n^{x_i}$.
	\end{claim}
	
	\begin{proof}
		Let us first observe that the inequality trivially holds for $x_i=0$, and it also easily holds for $x_i = 1$: $|\binom{X_i\cdot n}{x_i} - \binom{f_i\cdot n}{x_i}| = |X_i\cdot n - f_i\cdot n| \le \delta n \le 4 \delta n^{x_i}$. Therefore, let us assume now that $x_i \ge 2$.
		
		Let us recall a binomial identity: $\binom{L}{M} = \sum_{K=M-1}^{L-1} \binom{K}{M-1}$, which gives for $M \le J \le L$ the following: $\binom{L}{M} = \binom{J}{M} + \sum_{K=J}^{L-1} \binom{K}{M-1}$. Using this identity, that $f_i \le 1$, and $x_i \ge 2$, we obtain,
		\begin{eqnarray*}
			\binom{X_i\cdot n}{x_i}
			& \le &
			\binom{f_i\cdot n + \lceil\delta n\rceil}{x_i}
			=
			\binom{f_i\cdot n}{x_i} +
			\sum_{j=1}^{\lceil\delta n\rceil}\binom{f_i\cdot n+j-1}{x_i-1}
			\\
			& \le &
			\binom{f_i\cdot n}{x_i} +
			\lceil\delta n\rceil \cdot \binom{f_i\cdot n + \lceil\delta n\rceil - 1}{x_i-1}
			\\
			& \le &
			\binom{f_i\cdot n}{x_i} +
			2 \delta n (f_i\cdot n + \delta n)^{x_i-1}
			=
			\binom{f_i\cdot n}{x_i} + 2 \delta n ((1 + \delta)n)^{x_i-1}
			\\
			& = &
			\binom{f_i\cdot n}{x_i} + 2 \delta (1 + \delta)^{x_i-1} n ^{x_i}
			\le
			\binom{f_i\cdot n}{x_i} + 4 \delta n^{x_i}
			\enspace,
		\end{eqnarray*}
		%
		where in the last inequality, we used the fact that $(1+\delta)^{x_i-1} \le (1+\delta)^{k} \le 2$.
		
		Similarly, if $f_i\cdot n \ge \lceil\delta n\rceil + k$, we have $f_i\cdot n \ge \lceil\delta n\rceil + x_i$, and we obtain,
		\begin{eqnarray*}
			\binom{X_i}{x_i}
			& \ge &
			\binom{f_i\cdot n - \lceil \delta n \rceil}{x_i}
			=			
			\binom{f_i\cdot n}{x_i} -
			\sum_{j=1}^{\lceil\delta n\rceil}\binom{f_i\cdot n - j}{x_i-1}
			\\
			& \ge &
			\binom{f_i\cdot n}{x_i} -
			\lceil\delta n\rceil \binom{f_i\cdot n}{x_i-1}
			\ge
			\binom{f_i\cdot n}{x_i} - 2 \delta n \binom{n}{x_i-1}
			\\
			& \ge &
			\binom{f_i\cdot n}{x_i} - 2 \delta n \cdot n^{x_i-1}
			=
			\binom{f_i\cdot n}{x_i} - 2 \delta n^{x_i}
			\enspace.
		\end{eqnarray*}
		
		On the other hand, if $f_i\cdot n \le \lceil\delta n\rceil + k$, we note that $\binom{f_i\cdot n}{x_i} \le \binom{\lceil \delta n \rceil + k}{x_i} \le (\lceil \delta n \rceil + k)^{x_i} \le (2 \delta n)^{x_i} \le 4 \delta n^{x_i}$, where the third inequality follows from the fact that $n\geq n_1$ and that $n_1$ is a sufficiently large constant. Therefore since $\binom{X_i}{x_i} \ge 0$, we have $\binom{X_i}{x_i} \ge \binom{f_i\cdot n}{x_i} - 4 \delta n^{x_i}$.
		
		Now we can combine all the bounds above and obtain that for $x_2 \ge 2$, the following holds,
		\begin{displaymath}
		\binom{f_i\cdot n}{x_i} - 4 \delta n^{x_i}
		\le
		\binom{X_i}{x_i}
		\le
		\binom{f_i\cdot n}{x_i} + 4 \delta n^{x_i}
		\enspace,
		\end{displaymath}
		what yields the claim.
	\end{proof}
	
	Next, consider any $F = \{\Gamma_1,\dots,\Gamma_{k}\}$ and the corresponding frequencies $x_1, \dots, x_N$. Note that there are at most $k$ indices $i$ with $x_i > 0$, 
	and that $x_1 + \dots + x_N = k$. Let $\II = \{i: x_i>0, 1 \le i \le N\}$ and thus $|\II| \le k$ and $\prod_{i \in \II} n^{x_i} = n^k$. We have the following auxiliary claim.
	
	\begin{claim}
		\label{claim:aux1}
		For any $i$, conditioned on $|X_i - f_i| \le \delta$, the following inequalities hold:
		\begin{eqnarray*}
			\prod_{i \in \II}\left(\binom{f_i\cdot n}{x_i} +
			4 \delta n^{x_i}\right)
			& < &
			\prod_{i \in \II}\binom{f_i\cdot n}{x_i} + 4 \delta 2^k n^k
			\enspace,
			\\
			\prod_{i \in \II}\left(\binom{f_i\cdot n}{x_i} -
			4 \delta n^{x_i}\right)
			& > &
			\prod_{i \in \II}\binom{f_i\cdot n}{x_i} - 4 \delta 2^k n^k
			\enspace.
		\end{eqnarray*}
	\end{claim}
	
	\begin{proof}
		For any $i \in \II$, we let $y_{i,0} = \binom{f_i\cdot n}{x_i}$ and $y_{i,1} = 4 \delta n^{x_i}$. Then
		\begin{eqnarray*}
				\prod_{i \in \II}
				\left(
				\binom{f_i\cdot n}{x_i} + 4 \delta n^{x_i}
				\right)
			& = &
			\prod_{i \in \II}(y_{i,0} + y_{i,1})
			=
			\sum_{i \in \II, j_i \in \{0,1\}} \prod_{i \in \II} y_{i,j_i}
			\\
			& = &
			\prod_{i \in \II} y_{i,0} +
			\sum_{\substack{i \in \II, j_i \in \{0,1\},\\
					\textrm{there exists $j_i=1$}}}\prod_{i\in \II}y_{i,j_i}
			\\
			& = &
			\prod_{i \in \II} \binom{f_i\cdot n}{x_i} +
			\sum_{\substack{i \in \II, j_i \in \{0,1\},\\
					\textrm{there exists $j_i=1$}}}\prod_{i\in \II}y_{i,j_i}
			\enspace.
		\end{eqnarray*}
		Now note that for any $i \in \II$,  $y_{i,0} = \binom{f_i\cdot n}{x_i} \le n^{x_i}$. Therefore, for any sequence $\{j_i\}_{i \in \II}$ with at least one element equal to $1$, we have the following bound $\prod_{i \in \II} y_{i,j_i} \le 4 \delta \prod_{i \in \II} n^{x_i} = 4 \delta n^k$. Since the total number of such indices is $2^k-1 < 2^k$, we have
		\begin{eqnarray*}
			\prod_{i \in \II}\binom{f_i\cdot n}{x_i} +
			\sum_{\substack{i \in \II, j_i\in\{0,1\},\\
					\textrm{there exists $j_i=1$} }}\prod_{i\in \II}y_{i,j_i}
			< \
			\prod_{i \in \II}\binom{f_i\cdot n}{x_i} + 4 \delta n^k \cdot 2^k
			\enspace,
		\end{eqnarray*}
		which completes the proof of the first inequality. The proof of the second inequality is analogues.
	\end{proof}
	
	Using Claims \ref{claim:binomial} and \ref{claim:aux1}, we can prove the following relation between $\appr(F)$ and $\pr(F)$.
	
	\begin{claim}
		\label{claim:appr-bound}
		If $|X_i - f_i| \le \delta$ for every $i$, then $|\appr(F) - \pr(F)| \le 4 \delta (2k)^k$ for every $F \in \FF_n$.
	\end{claim}
	
	\begin{proof}
		Let $F = \{\Gamma_1, \dots, \Gamma_k\} \in \FF_n$. By Claims \ref{claim:binomial} and \ref{claim:aux1}, we have
		\begin{eqnarray*}
			\appr(\Gamma_1 \dots \Gamma_k)
			& = &
			\frac{\prod_{i \in \II}\binom{X_i\cdot n}{x_i}}{\binom{n}{k}}
			\le
			\frac{\prod_{i \in \II}\left(\binom{f_i\cdot n}{x_i} + 4 \delta n^{x_i}\right)}{\binom{n}{k}}
			<
			\frac{\prod_{i \in \II}\binom{f_i\cdot n}{x_i} + 4\delta 2^k n^k}{\binom{n}{{k}}}
			\\
			& \le &
			\pr(\Gamma_1, \dots, \Gamma_k) + 4 \delta (2k)^k
			\enspace,
		\end{eqnarray*}
		where the last inequality follows from that $\binom{n}{k} \ge (\frac{n}{k})^k$. Similarly, by Claims \ref{claim:binomial} and \ref{claim:aux1}, we have,
		\begin{eqnarray*}
			\appr(\Gamma_1, \dots, \Gamma_k)
			& \ge &
			\frac{\prod_{i \in \II} \left(\binom{f_i\cdot n}{x_i} - 4 \delta n^{x_i}\right)} {\binom{n}{k}}
			\ge
			\frac{\prod_{i \in \II}\binom{f_i\cdot n}{x_i} - 4 \delta 2^k n^k}{\binom{n}{k}}
			\\
			& \ge &
			\pr(\Gamma_1, \dots, \Gamma_k) - 4 \delta (2k)^k
			\enspace.
		\end{eqnarray*}
	\end{proof}
	
	%
	
Now consider the case that $G$ satisfies $\Pi$. Then, by Lemma \ref{lemma:constant_query}, with probability at most $\frac{5}{12}$, the subgraph $S_{G,k}$ spanned by the $k$-discs of $k$ vertices that are sampled uniformly at random without replacement is isomorphic to some member in $\FF_n$, that is, $\sum_{F \in \FF_n} \pr(F) \le \frac{5}{12}$. Therefore, by Claim \ref{claim:appr-bound}, we have,
	\begin{eqnarray*}
		\sum_{F \in \FF_n} \appr(F)
		& < &
		\sum_{F \in \FF_n} \pr(F) + \sum_{F \in \FF_n} 4 \delta (2k)^k
		\le
		\sum_{F \in \FF_n} \pr(F) + N_{d,k}^k \cdot 4 \delta (2k)^k
		\\
		& \le &
		\frac{5}{12} + \frac{1}{12}
		=
		\frac12
		\enspace.
	\end{eqnarray*}
	
	Similarly, by Lemma \ref{lemma:constant_query}, if $G$ is $\varepsilon$-far from satisfying $\Pi$, then with probability at least $\frac{7}{12}$, the $k$-discs rooted at $k$ vertices that are sampled uniformly at random span a subgraph in $\FF_n$. Hence, Claim \ref{claim:appr-bound} gives
	\begin{eqnarray*}
		\sum_{F \in \FF_n} \appr(F)
		& \ge &
		\sum_{F \in \FF_n} \pr(F) - \sum_{F \in \FF_n} 4 \delta (2k)^k
		\ge
		\sum_{F \in \FF_n} \pr(F) - N_{d,k}^k \cdot 4 \delta (2k)^k
		\\
		& \ge &
		\frac{7}{12} - \frac{1}{12}
		=
		\frac12
		\enspace.
	\end{eqnarray*}
	
	These inequalities conclude the analysis of our algorithm and the proof of Theorem \ref{thm:const_query_testable_property}.
\end{proof}

\section{Constant-Time Approximation Algorithms}
As we mentioned in the introduction, to simulate any constant-time algorithm that is independent of the labeling of the vertices, and accesses the graph by sampling random vertices and exploring neighborhoods (or $k$-discs for some $k$) of these vertices, it suffices to have the distribution of $k$-disc types. Now we explain slightly more about this simulation and sketch the proof of  Theorem~\ref{thm:constant_time}.
In order to approximate the size of the solution of an optimization problem (e.g., maximum matching, minimum vertex cover), it has been observed by Parnas and Ron~\cite{PR07:sublinear_distributed} that it suffices to have efficient oracle $\Oracle_S$ access to a solution $S$. This is true since one can attain a good estimator for the size of $S$ by sampling a constant number of vertices, performing corresponding queries to the oracle $\Oracle_S$ and then returning the fraction of vertices that belong to $S$ based on the returned answers from $\Oracle_S$.
Nguyen and Onak~\cite{NO08:constant} implemented such an oracle via an elegant approach of locally simulating the classical greedy algorithm. In particular, they showed the following result.

\begin{lemma}[\cite{NO08:constant}]\label{lemma:onak_oracle}
	There exist $q=q(\varepsilon,d)$, an oracle $\Oracle_M$ to a maximal matching $M$, and an algorithm that queries $\Oracle_M$ about all the edges incident to a set of $s=O(1/\varepsilon^2)$ randomly sampled vertices and with probability at least $2/3$, returns an estimator that is $(1,\varepsilon n)$-approximation of the size of $M$, and each query to $\Oracle_M$ performs at most $q$ queries to the adjacency list of the graph.   
\end{lemma}



A key observation is that the algorithm in Lemma~\ref{lemma:onak_oracle} can be viewed as first sampling $s$ $q$-discs from the graph and then perform $\Oracle_M$ queries on each of these $q$-discs. It is easy to see that with high probability $0.99$, all these $q$-discs are disjoint. Furthermore, the answer of the above oracle only depends on the structure of the corresponding neighborhood of the starting vertex $v$ and the random ordering of the edges belonging to this neighborhood. 
Now we can approximate the size of a maximal matching in the random order streaming model as follows: we first invoke Algorithm~\ref{alg:approx_k_disc} to get an estimator for the distribution of
$q$-discs. Then we can simulate the oracle on this distribution.
\section*{Acknowledgment}
We would like to thank G. Cormode, H. Jowhari for helpful discussions.

\bibliographystyle{alpha}
\bibliography{random_stream}





\end{document}